\documentclass[12pt,a4paper]{amsart}
\usepackage{amsmath,amssymb,amsfonts,amsthm}
\usepackage[hypertex]{hyperref} 
\usepackage{hyperref} 
\usepackage{comment}

\date{26 May (8 June) 2015}

\author{Theodore~Voronov}
\address{School of Mathematics, University of Manchester, Manchester, M60 1QD, UK\\
{\hphantom{hh}Dept. of Quantum Field Theory, Tomsk State University, Tomsk, 634050, Russia}}
\email{theodore.voronov@manchester.ac.uk}

\title[Thick morphisms   and oscillatory integral operators]{Thick morphisms of supermanifolds and oscillatory integral operators}

\newtheorem{theorem}{Theorem}

\theoremstyle{definition}
\newtheorem*{definition}{Definition}

\def\co{\colon\thinspace}

\newcommand{\der}[2]{{\frac{\partial {#1}}{\partial {#2}}}}

\newcommand{\fun}{C^{\infty}}

\newcommand{\e}{\varepsilon}

\newcommand{\tto}{{\linethickness{2pt}
		  \,\begin{picture}(1,0)
                   \put(0,0.26){\line(1,0){0.95}}
                   \put(0,0){$\boldsymbol{\rightarrow}$}
                  \end{picture}
                  }\,
}

\newcommand{\ttoq}{\tto_q}

\begin{document}
\begin{abstract} We show that thick morphisms (or microformal morphisms) between smooth (super)manifolds, introduced by us before, are classical limits of `quantum thick morphisms' defined here as particular oscillatory integral operators on functions.
\end{abstract}

\maketitle

In~\cite{tv:nonlinearpullback,tv:microformal} we introduced nonlinear pullbacks of functions  with respect to `microformal' or `thick' morphisms of   (super)manifolds, which generalize ordinary smooth maps. By definition, such a morphism is a formal canonical relation between the cotangent bundles of a special kind, namely, specified by a generating function depending on position coordinates on the source and momentum coordinates on the target. This function is seen as a power expansion near the zero section. Thick morphisms form a formal category; that means that the composition law for the generating functions is a formal power series. Likewise, the pullback of a function w.r.t. a thick morphism is given by a formal power series whose terms are nonlinear differential operators. There is a parallel construction based on anticotangent bundles yielding nonlinear pullbacks of odd functions (the former construction applies to even functions). Our main application was to   $L_{\infty}$-morphisms between homotopy Schouten or Poisson algebras of functions. Another application was the construction of an `adjoint operator' for nonlinear maps of vector bundles. (Thick morphisms in the even version are close to symplectic micromorphisms of Cattaneo--Dherin--Weinstein, see~\cite{cattaneo-dherin-weinstein:one} and subsequent works,  defined as germs of canonical relations between germs of symplectic manifolds at Lagrangian submanifolds; analogs of our pullbacks do not arise in such a setting. See further discussion of this in~\cite{tv:microformal}.)

We show here that thick morphisms of microformal geometry can be seen as the classical limit of certain `quantum thick morphisms', which are given by oscillatory integral operators of a particular kind. Let us point out that oscillatory integral operators (and Fourier integral operators) are well known, as well as well known is their connection with canonical relations between cotangent bundles. Roughly, each such relation defines a class of Fourier integral operators (see, e.g., \cite{guillemin-sternberg:semiclassical}). We, however, consider a very special integral operator in this class, generalizing the operator of pullback w.r.t. a smooth map. It is defined by a `quantum' version of a generating function specifying a thick morphism. `Quantum' here mean   depending on $\hbar$. The action of such operators on oscillatory wave functions in the classical limit exactly reproduces the nonlinear pullback of~\cite{tv:nonlinearpullback,tv:microformal}. The same holds true for the composition of our operators: in the classical limit it reduces to the composition of thick morphisms.

We have departed from the following observation. Consider thick morphisms of a (super)manifold $M$ to itself; and let them be invertible. We arrive at a formal (super)group of `thick diffeomorphisms' of $M$. What is its Lie (super)algebra? 
\begin{theorem} The Lie superalgebra of the formal supergroup of the  thick diffeomorphisms  of a supermanifold $M$ can be identified with the Lie superalgebra $\fun(T^*M)$ w.r.t. the canonical Poisson bracket. The infinitesimal action of $\fun(T^*M)$ (corresponding to the nonlinear pullbacks) is the (nonlinear) `Hamilton--Jacobi action' $f(x)\mapsto f(x)+\e H(x,\der{f}{x})$ on even functions.
\end{theorem}
\begin{proof}
We check the last statement. The generating function of a thick morphism $\Phi\co M\tto M$ close to the identity is given by $S(x,p)=x^ap_a+\e H(x,p)$. The pullback $\Phi^*$ sends $f(x)$ to $f(y)+S(x,p)-y^ap_a=f(y)+\e H(x,p)$, where $y$ and $p$ are to be found from $y^a=\der{S}{p_a}(x,p)=x^a+\e \der{H}{p_a}(x,p)$, $p_a=\der{f}{x^a}(y)$, hence $f(x)\mapsto f(x)+\e H(x,\der{f}{x})$ as claimed.
\end{proof}
We can now see two things. First, as $\fun(T^*M)$ is the Lie algebra of the group of canonical transformations of $T^*M$, this group acts, at least infinitesimally, on even functions on $M$. This is reminiscent of the spinor representation of the orthogonal and symplectic groups, and we expect the existence of a link with it. (We also see that the formal supergroups of thick diffeomorphisms of $M$ and formal canonical transformations of $T^*M$ should be isomorphic.) Secondly, as the Hamilton--Jacobi equation is the classical analog of the Schr\"{o}dinger equation, there should be a `Schr\"{o}dinger' or `quantum' version of thick morphisms and their action on functions. This is exactly given by the oscillatory integral operators considered below.

\begin{definition} A \emph{quantum thick morphism} $\hat\Phi\co M_1\ttoq M_2$ is given by an integral operator (a `quantum pullback') sending functions on $M_2$ to functions on $M_1$  and denoted $\hat\Phi^*$ by the formula
\begin{equation}
    (\hat\Phi^* w)(x)=\frac{1}{(2\pi\hbar)^{n_2}} \int_{T^*M_2} DyDq \,\, e^{\frac{i}{\hbar}(S_{\hbar}(x,q)-y^iq_i)}\,w(y)\,.
\end{equation}
Here $n_2=\dim M_2$. (For simplicity of notation, the formula is written for ordinary manifolds; modification for the supercase is obvious.)
\end{definition}

The function $S_{\hbar}(x,q)$ in this formula is analogous to the generating function $S(x,q)$ specifying thick morphisms~\cite{tv:nonlinearpullback,tv:microformal}. It a power series in $q_i$ and in $\hbar$. 

\begin{theorem} On the phases of the oscillatory wave functions, the quantum pullback $\hat\Phi^*$ in the classical limit $\hbar\to 0$, induces the nonlinear pullback $\Phi^*$ as defined in~\cite{tv:nonlinearpullback,tv:microformal} w.r.t. the thick morphism specified by the generating function $S(x,q)=S_0(x,q)$,
\end{theorem}

(A similar statement holds for the composition of quantum thick morphisms. The integral formula obtained in the zeroth order in $\hbar$ gives the composition law for  generating functions~\cite{tv:microformal}, but contains higher corrections in $\hbar$. This explains why one should assume a dependence on $\hbar$ in $S_{\hbar}(x,q)$ from the beginning: even if there is not, the `quantum composition law' will produce such a dependence for the generating function of the composition. Compare star-product in deformation quantization.)
\begin{proof} Consider a function on $M_2$ of the form $w(y)=e^{\frac{i}{\hbar}g(y)}$. We arrive at the integral
\begin{equation}\label{qpback}
    (\hat\Phi^* w)(x)=\frac{1}{(2\pi\hbar)^{n_2}} \int_{T^*M_2} DyDq \,\, e^{\frac{i}{\hbar}(g(y)+S_{\hbar}(x,q)-y^iq_i)}\,.
\end{equation}
By the stationary phase method, in the limit $\hbar\to 0$, the integral will be asymptotically equal to $u(x)=e^{\frac{i}{\hbar}f(f)}$ with $f(x)$ equal to the value of the phase function in~\eqref{qpback} at the stationary point (as function of $y,q$). Differentiating $g(y)+S_{\hbar}(x,q)-y^iq_i$ and setting $\hbar$ to zero, we obtain 
\begin{align*}
    \der{}{y^j}\bigl(g(y)+S_{0}(x,q)-y^iq_i)\bigr) & \equiv \der{g}{y^j}(y)-q_j=0\,,\\
    \der{}{q_j}\bigl(g(y)+S_{0}(x,q)-y^iq_i)\bigr) & \equiv \der{S_{0}}{q_j}(x,q)  -y^j =0\,,
\end{align*}
which are exactly the equations 
\begin{align*}
     q_j&=\der{g}{y^j}(y)\,,\\
     y^j&= \der{S}{q_j}(x,q) 
\end{align*}
in the definition of the nonlinear pullback $\Phi^*$. Hence $f=\Phi^*[g]$\,.
\end{proof}

Bearing in mind our initial motivation, constructing $L_{\infty}$-morphisms of algebras of functions, we expect to obtain corresponding `quantum versions' from quantum thick morphisms. We hope to elaborate this elsewhere.


\def\cprime{$'$}  

\end{document}